\documentclass[letterpaper, 10 pt, conference]{ieeeconf}  

\IEEEoverridecommandlockouts                              

\usepackage{graphics} 
\usepackage{epsfig} 
\usepackage{mathptmx} 
\usepackage{times} 
\usepackage{amsmath} 
\usepackage{amssymb}  
\usepackage{algorithm}
\usepackage{algorithmic}

\usepackage[utf8]{inputenc}
\usepackage[english]{babel}

\newtheorem{theorem}{Theorem}[section]

\newtheorem{lemma}[theorem]{Lemma}

\usepackage[noadjust]{cite}
\usepackage{hyperref}
\usepackage{multirow}
\usepackage{booktabs}
\usepackage{caption}
\usepackage{xcolor}

\renewcommand{\baselinestretch}{0.94}
\title{\LARGE \bf Stochastic Traveling Salesperson Problem with Neighborhoods for Object Detection}

\author{Cheng Peng, Minghan Wei, and Volkan Isler
\thanks{This work was supported by NSF grant number 1849107.}
\thanks{The authors are with the Department of Computer Science and Engineering,
        University of Minnesota, 200 Union Street SE, Minneapolis, MN, US, 55108
        {\tt\small peng0175, weixx526, isler@umn.edu}}%
\thanks{The authors would like to thank Shan Su for valuable discussions, and Nicolai Hani and Selim Engin for their help with data generation.}
}

\usepackage{color}

\begin{document}

\maketitle
\thispagestyle{empty}
\pagestyle{empty}

\begin{abstract}
We introduce a new route-finding problem which considers perception and travel costs simultaneously. Specifically, 
we consider the problem of finding the shortest tour such that all objects of interest can be detected successfully. To represent a viable detection region for each object, we propose to use an entropy-based viewing score that generates a diameter-bounded region as a viewing neighborhood. We formulate the detection-based trajectory planning problem as a stochastic traveling salesperson problem with neighborhoods and propose a center-visit method that obtains an approximation ratio of $O(\frac{D_{max}}{D_{min}})$ for disjoint regions. For non-disjoint regions, our method provides a novel finite detour in 3D, which utilizes the region's minimum curvature property. Finally, we show that our method can generate efficient trajectories compared to a baseline method in a photo-realistic simulation environment.
\end{abstract}

\section{Introduction}
Coverage path planning for high-quality texture mapping and 3D reconstruction has been an active research topic for decades. For smaller objects, one of the main objectives is to recognize/reconstruct the object with a small number of views.
For larger scenes/structures, trajectory planning is needed regarding the autonomous agent's motion and energy limitations. One commonality for both applications is that path planning is optimized for all visible regions. In basic settings, all surfaces of a scene and objects are treated equally to gain sufficient views for reconstruction or texture mapping~\cite{peng2018view, peng2019view, vasquez2014volumetric, bircher2016receding, isler2016information, fan2016automated, scott2003view}. 

However, there are applications, such as inspection for a particular region, that require only a subset of regions to be covered. Trajectory planning for all surfaces can be  redundant and inefficient. For tasks such as object detection and segmentation, the state-of-the-art detection methods~\cite{yolov3, matterport_maskrcnn_2017} can successfully identify objects from a single view. So we only need to plan for a subset of surfaces/objects that are semantically important. Since the detection and segmentation process can be stochastic~\cite{zhu2017target}, the planning method may also need to consider the uncertainty associated with object states such as locations and orientations.

One trend of work assumes the prior information follows certain probabilistic distributions called Partial Observable Markov Decision Process (POMDP). The goal is to find a trajectory that maximizes the expected long-term gain for detecting all objects. However, problems such as finding the shortest path through multiple locations still pose extreme challenges due to the too large solution to be learned efficiently. There is no performance guarantee compared to optimal solutions. Therefore, works related to path-finding are mostly constrained to single target navigation~\cite{zhu2017target}. 

In this paper, our objective is to plan the shortest trajectory such that all objects of interest in the scene can be detected successfully, where the locations of those objects are given in advance. Such cases arise when we have the locations of the objects and would like to find their particular attributes such as types, brands, and colors. 
Note that in our formulation, specific observation locations are not given and must be explicitly computed (Figure~\ref{fig:traj_detection}).
We formulate this detection task with prior location information into a 3D stochastic Traveling Salesperson Problem with neighborhoods (TSPN) where the neighborhoods are constructed by non-convex regions called diameter-bounded regions (Section~\ref{sec:diameter_bounded}). Once such a region is visited, the object inside can be detected with high confidence. The diameter-bounded region around an object intends to directly map a viewing pose to a detection score. Due to the lack of sufficient data representing objects from all possible viewing poses, we also propose an entropy-based viewing score to find the 3D diameter-bounded region for each object.
Based on this formulation, we present a polynomial time approximation method for the 3D stochastic TSPN problem with $O(\frac{D_{max}}{D_{min}})$ factor for disjoint regions and ($O(\frac{D_{max}^2}{D_{min}^2})$) for non-disjoint regions, where $D_{max}$ and $D_{min}$ are the diameters of the expected prior distribution region. In Section~\ref{sec:expr}, we show that the diameter ratio in the approximation factors is small and not computationally prohibited.

\begin{figure}
    \centering
	\includegraphics[width=0.4\textwidth]{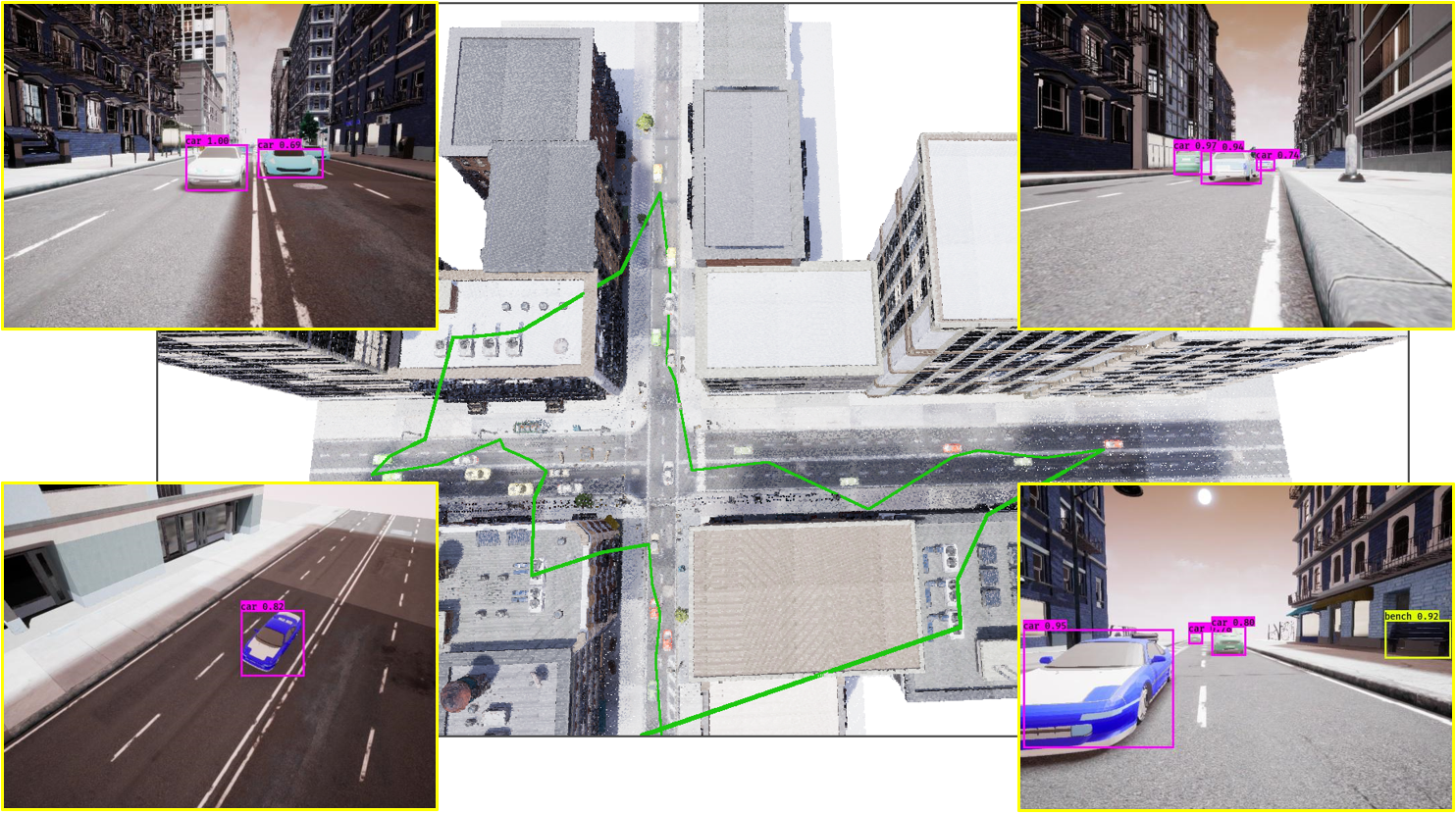}
    \caption{The trajectory of an aerial vehicle for observing a set of cars in the Unreal Engine simulator~\cite{unrealengine}. Images are captured along the trajectory toward the objects of interest. Detection results using yolo-v3~\cite{yolov3} network are shown here.}
    \label{fig:traj_detection}
\end{figure}

\section{Related work}
There is a large amount of literature on shortest-tour planning for a set of neighborhoods. We briefly review related work on the Traveling Salesperson Problem (TSP) and its variant called Stochastic TSP. 
\subsection{Traveling Salesperson Problem}
Traveling Salesperson Problem (TSP) is to find the shortest tour for a set of locations such that each location is visited exactly once. If the distance metric is Euclidean distance, it is called the Euclidean TSP. If an arbitrary metric is used, the problem is called metric TSP~\cite{arora2003approximation}, which is known to be NP-Hard. Numerous approximation algorithms have been proposed for Euclidean TSP. One of the earliest methods provided an approximation ratio 1.5~\cite{christofides1976worst} for Euclidean TSP which traverses a minimum spanning tree of the locations. For metric TSP, a Polynomial Time Approximation Scheme (PTAS) of $1+\epsilon$ for some $\epsilon > 0$ was discovered by Arora~\cite{arora1998polynomial} and Mitchell~\cite{mitchell1999guillotine} where the dimension of the problem need to be constant.

For Traveling Salesperson Problem with Neighborhoods (TSPN), the objective is to visit a set of regions, where the Euclidean version is APX-Hard~\cite{de2005tsp, safra2006complexity, papadimitriou1977euclidean}. The final trajectory must visit one point within the region. This problem was initially studied by Arkin and Hassin~\cite{arkin1994approximation}, where constant factor approximations are provided where the neighborhoods are parallel unit segments, translates of a polygonal region, and circles. Dumitrescu and Mitchell~\cite{dumitrescu2003approximation} provided a constant factor approximation for connected regions with nearly the same diameters in the plane. Elbasio et al.~\cite{elbassioni2005approximation} generalized to disjoint fat regions with $9.1\alpha + 1$ approximation, where $\alpha$ is a measure of fatness that does not constrain the convexity or diameter of the region. For connected regions in higher dimensions, Dumitrescu and Toth~\cite{dumitrescu2016traveling} gave a $O(1)$ approximation algorithm for a set of $n$ hyperplanes and similarly for a set of unit balls in $\mathbb{R}^d$, where $d$ is a constant. 

For more specific covering tasks, Plonski and Isler~\cite{plonski2019approximation} introduced right circular cones as neighborhoods to represent a camera's field of view. Their method provides an approximation factor of $O(1+\log(h_{max}/h_{min}))$, where $h_{max}$ and $h_{min}$ is the maximum and minimum cone height. Stefas et al.~\cite{stefas2018approximation} later extended the problem to cones with different bisector orientations and their method provides an approximation factor of $O(\frac{1+\tan{\alpha}}{1-\tan{\epsilon}\tan{\alpha}}(1+\log(h_{max}/h_{min})))$, where $\epsilon$ is the cone orientation and $\alpha$ is the apex angle.

\subsection{Stochastic Traveling Salesman Problem}
When the neighborhood shapes are stochastic, classic methods no longer apply. To model this problem, a TSPN variant called Stochastic TSPN is introduced that solves for an expected optimal solution or the worst case bound.

Kamousi and Suri~\cite{kamousi2013euclidean} proposed to assign neighborhoods of disks with the radius sampled from a distribution. This problem aims to model applications such as data mule~\cite{tekdas2012efficient} where the distance of communication is not certain. The resulting approximation method is compared with the expected optimal trajectory. With $n$ stochastic disks, their method achieves an approximation factor of $O(\log \log n)$. Bertsimas and Jaillet~\cite{jaillet1988priori} studied a setting where the input locations have some activation probability. Instead of studying the expected solution, Citovsky et al.~\cite{citovsky2017tsp} studied an adversarial case such that the worst-case trajectory length is bounded. However, there are no existing efficient algorithms for more general shapes such as non-convex 3D objects.

Blum et al.~\cite{blum2007approximation} formulated TSP with rewards as an orienteering problem. The goal is to collect as much price as possible subject to a limited path length. When the current state and action are associated with probabilistic distributions, it can therefore be formulated as a Markov Decision Process (MDP)~\cite{blum2007approximation}. When the observation of the current state is not complete, the problem becomes a Partial Observable Markov Decision Process (POMDP), which is normally solved by updating the Bellman equation. To encompass larger dimensions of the state, deep reinforcement learning~\cite{zhu2017target, mnih2015human} was introduced to explore the large state space and learned to approximate the optimal policy with neural networks.
\section{Problem Formulation}
We are given a set of objects $x \in \mathcal{X}$ with their states. The state of each $x$ is denoted as $s_x$, which can be the location and orientation. We also define an attribute function $Attr(x)$ that we wish to predict such as semantic information or metric information. For semantic information, we can treat $Attr(x)$ as the object label. For metric information, we can treat $Attr(x)$ as object location and orientation. In order to predict $Attr(x)$, we need to take a measurement $I(v, s_x)$ where $v$ is the camera pose $v \in \mathbb{SE}^3$.

We are given a function  $F(I(v, s_x, B_x))$ that outputs the estimate of $Attr(x)$ as $\hat{y_x} = F(I(v, s_x, B_x))$ where the ground truth attribute is $y_x^* = Attr(x)$.

To find attributes of all objects in $\mathcal{X}$, we define the corresponding trajectory as $J = \{v_1, v_2, ..., v_{|\mathcal{X}|}\}$, where the trajectory length is denoted as $|J| = \sum_{i=1}^{|J|-1} ||v_{i+1} - v_i||_2$. We take measurements at each $v$ for the corresponding object $x$.

The objective is to minimize the trajectory length $|J|$ such that the difference between the prediction and ground truth of all object attributes is bounded.

\begin{equation}
\label{eq:prob-nonstochastic}
\begin{split}
    & \min \hspace{0.5cm} (len(J)) \\
s.t. \hspace{1cm} & dis(y_x^*, \hat{y_x}) \leq T, \forall x \in \mathcal{X}
\end{split}
\end{equation}

where $dis(\cdot, \cdot)$ is a distance function and $T$ is a distance threshold to upper bound the deviation.

\subsection{Prediction uncertainty}
In practice, the measurement of $\hat{y}_x$ can be probabilistic. Let $\hat{y}_x = P(\hat{y}_x | v, s_x)$ be the probabilistic distribution of measurement results for $s_x$ at a camera pose $v$.
We can define a function: $G(v, s_x) = \int P( \hat{y}_x | v, s_x)(dis(\hat{y}_x - y)) d\hat{y}_x
$. $G(v, s_x)$ returns the expected error of the prediction for $(v)$ and $s_x$.

\subsection{Expected Prediction Region}

The function $G(v, s_x)$ can be considered as a heat map around the object center $C(x)$, which shows the expected prediction error. It is desired that the measurement results for the attribute have high confidence. Therefore, for each object state $s_x$, we define its expected detection region to be $R(s_x) = \{v | G(v, s_x)\leq T\}$, where $T$ is the distance threshold in Equation~\ref{eq:prob-nonstochastic}. In $R(s_x)$, we can obtain a prediction with the expected error less than $T$.

With the above definitions, we could rewrite the {\em stochastic traveling salesperson problem with Neighborhoods problem} as:
\begin{equation}
\label{eq:prob}
\begin{split}
     & \min \hspace{0.5cm} (|J|) \\
s.t. \hspace{1cm} & J\bigcap R(s_x) \neq \emptyset, \forall x \in \mathcal{X}
\end{split}
\end{equation}
In Equation~\ref{eq:prob}, our goal is to minimize the tour length, under the condition that the expected detection error for each object is less than $T$.

In this paper, we make the following two assumptions of the geometry of $R(x)$ based on our experiments in Section~\ref{sec:expr}. First, we assume that $ R(x)$ is a simply connected 3D region. Second, we assume that $R(x)$ is a diameter-bounded region, which can be non-convex (Section~\ref{sec:diameter_bounded}). In Section~\ref{sec:expr}, we generalize detection scores for $60$ objects based on our geometric interpretation of detection accuracy.

\section{Method}
In this section, we propose an algorithm that finds an efficient trajectory to obtain the attribute for each object. For each object denoted as $x \in \mathcal{X}$, the detection range is stochastic. 
To gain insights into the problem, we first propose to solve an offline expected case with known region locations and orientations for each $x \in \mathcal{X}$. This result is then compared with the expected optimal trajectory.  In the first case, the problem becomes to find the shortest tour that visits all $R(x)$ for $x \in \mathcal{X}$. 
Then, we extend the solution for the online case where the orientations of $R(x)$ are unknown.
We make the following assumptions for the regions. The expected region shape $R(x)$ is diameter-bounded (Section~\ref{sec:diameter_bounded}). However, we do not constrain the convexity of $R(x)$. 

\subsubsection{Diameter-bounded region}\label{sec:diameter_bounded}
Before the discussion of the algorithm, we formally define the diameter-bounded region here. 
For each object $x$, we define a diameter-bounded region $R(x)$ with two diameter values $D_{min}(x)$ and $D_{max}(x)$ of $x$, where $D_{min}(x)$ is the minimum diameter of all inscribed circles of $R(x)$ and $D_{max}(x)$ is the maximum diameter of $R(x)$. The diameters are bounded on both ends as $D_{min} \leq D_{min}(x) \leq D_{max}(x) \leq D_{max}$ for all expected regions in $R(\mathcal{X})$. Another way to understand $D_{min}$ is that the curvature of every location on the region boundary is less or equal to $\frac{1}{D_{min}/2}$. The center of $R(x)$ is denoted as $C(x)$, which is defined as the geometric center of $R(x)$.

\subsection{Offline disjoint expected case}

In the offline case, the expected detection region $R(x)$ of each object $x$ is given. In Section~\ref{sec:expr} we show how we generate the regions. The objective is to find the shortest tour, denoted as $E[J^*]$, so that the expected prediction error for each $x$ is less than $T$, as formulated in Equation~\ref{eq:prob}. Note that we use $E[J^*]$ instead of $J^*$ since in $R(x)$, the prediction error is still probabilistic and the expected error is less than $T$

\begin{algorithm}
\caption{Center-Visit}
\hspace*{\algorithmicindent} \textbf{Input:  $s_0 \in \mathbb{R}^3$, $R(\mathcal{X})$} \\
\hspace*{\algorithmicindent} \textbf{Output: $J_t$} 
\begin{algorithmic}[1]\label{alg:center-visit}
\footnotesize
\STATE Compute a trajectory from method~\cite{arora1998polynomial} using region centers from $R(\mathcal{X})$ and assign the resulting visiting order as $P(\mathcal{X})$.
\STATE $J_t \leftarrow \{s_0\}$
\FOR{$x_i \in P(\mathcal{X})$}
\STATE Denote the closest point on $R(x_i)$ to the last point in $J_t$ as $s_i$.
\STATE Append $s_i$ to $J_t$.
\ENDFOR
\STATE Output $J_t$
\end{algorithmic}
\end{algorithm}

\begin{theorem}\label{thrm:center-visit}
Given a set of diameter bounded region $R(x)$, the length of the trajectory $J_t$ from Algorithm~\ref{alg:center-visit} is at most $O(\frac{D_{max}}{D_{min}})$ of $|E[J^*]|$. 
\end{theorem}

To prove this theorem, we first establish a connection between the number of objects $N = |\mathcal{X}|$ and $|E[J^*]|$. The main idea is to use the Minkowski sum of the optimal trajectory with a ball that will cover a portion of the regions $R(\mathcal{X})$. 
\begin{lemma}\label{lem:minkowski}
Given a set of disjoint diameter-bounded regions $R(x)$ with range $[D_{min}, D_{max}]$, the number of objects $N = |\mathcal{X}|$ is upper bounded as follows by the expected optimal trajectory $|E[J^*]|$.
\begin{equation}
    N \leq \frac{27}{20 D_{min}} (|E[J^*]|+2D_{min})
\end{equation}
\end{lemma}
\begin{proof}
We consider $E[J^*]$ and the Minkowski sum of a ball $P$ of radius $D_{min}$ sweeping along $E[J^*]$. Since $E[J^*]$ touches each region at least once, the center of $P$ must also be on the boundary of each region at least once when sweeping along $J^*$. When the center of $P$ is on the boundary of a region, the overlapped regions between $P$ and the region is at least $\beta \frac{4}{3}\pi\frac{D_{min}^3}{8}$\footnote{$\beta \geq \frac{5}{12}$, the equality holds when the surface of two unit sphere touches each others' centers.}. The total overlapped volume with the $N$ regions when $P$ follows $E[J^*]$, should be less or equal to the total volume $P$ sweeps. Therefore, we have:
\begin{align*}
        \beta N\frac{4}{3}\pi\frac{D_{min}^3}{8} &\leq (|E[J^*]|+2D_{min}) (\frac{D_{min}}{2})^2\pi \\
        N & \leq \frac{27}{20 D_{min}} (|E[J^*]|+2D_{min})
\end{align*}
where $N = |\mathcal{X}|$.
\end{proof} 

After obtaining the upper bound for the number of objects, we can start constructing the tour to visit each region, which is presented in Algorithm~\ref{alg:center-visit}. The main idea is that given the regions $R(\mathcal{X})$ are fixed and known in prior, adding detours to $E[J^*]$ to the centers of $R(\mathcal{X})$ is guaranteed to cover all centers. Therefore, it must be lower bounded by the optimal trajectory $J_c$ that visits the centers of each region. 

\begin{lemma}\label{lem:centervisit}
Given a set of disjoint regions $R(\mathcal{X})$,
the optimal trajectory $J_c$ that visits the centers of $\mathcal{X}$ is upper bounded as follows. 
$|J_c| \leq |E[J^*]| + ND_{max}$
\end{lemma}
\begin{proof}
To visit the center of each region, $E[J^*]$ needs to take additional detours of at most $ND_{max}$ to visit the centers of each region. $E[J^*]$ together with the detour to visit the region centers forms a tour that visits the center of every region. The length of this tour should larger or equal to $J_c$, since $J_c$ is the optimal tour that visits every region center.
\end{proof}

Given Lemma~\ref{lem:minkowski} and ~\ref{lem:centervisit}, we can therefore present Algorithm~\ref{alg:center-visit} and the evaluate the performance here. 

\begin{proof}
The optimal trajectory that visits a set of points can be approximated from~\cite{arora2003approximation} with a factor of $1+\epsilon$, where $\epsilon \in (0, 1]$. Therefore, by combining Lemma~\ref{lem:minkowski} and ~\ref{lem:centervisit}, we can derive the following bound for a trajectory $J_t$ that visits the region centers using method~\cite{arora2003approximation}.
$\frac{1}{1+\epsilon}|J_t| \leq |E[J^*]| + ND_{max} \leq (1+\epsilon)(\frac{27}{20}\frac{D_{max}}{D_{min}} + 1) E[|J^*|]$
\end{proof}


\subsection{Offline non-disjoint case}
For non-disjoint regions, the common approaches~\cite{peng2019view, plonski2019approximation, dumitrescu2017constant} are to extract a maximal independent set and to construct a detour that visits the peripherals of the current region. Similarly, we also compute a maximal independent set $MIS(\mathcal{X})$ using Algorithm~\ref{alg:mis}. A detour for each region in $MIS(\mathcal{X})$ is then computed to visit its overlapping neighborhoods. However, detours around diameter-bounded regions are proportional to the surface area that can be an infinite detour length.

In this section, we show that given the input geometry $R(\mathcal{X})$ has minimum curvature constraints, the detour can be finite and efficient the planning a 3D world.

\begin{theorem}\label{thrm:detour-visit}
Given a set of non-disjoint regions, there exists a trajectory $|J_d|$ with length at most $O(\frac{D_{max}^2}{D_{min}^2})$ of $|E[J^*]|$.
\end{theorem}

\begin{algorithm}
\footnotesize
\caption{Maximal-Independent-Set}
\hspace*{\algorithmicindent} \textbf{Input: $R(\mathcal{X})$} \\
\hspace*{\algorithmicindent} \textbf{Output: $MIS(\mathcal{X})$} 
\begin{algorithmic}[1]\label{alg:mis}
\STATE $MIS(\mathcal{X}) = \emptyset$
\STATE Sort $R(\mathcal{X})$ using $D_{max}$ for each region in ascending order
\WHILE{ $R(\mathcal{X}) \neq \emptyset$}
\STATE $R(x_i)$ = the region with the smallest $D_{max}$ in $R(\mathcal{X})$
\STATE $MIS(\mathcal{X}) = MIS(\mathcal{X}) \bigcup R(x_i)$
\STATE Remove all $R(x_j)$ from $R(\mathcal{X})$ that intersect $R(x_i)$, i.e. $R(x_j) \bigcap R(x_i) \neq \emptyset$
\ENDWHILE
\STATE Output $MIS(\mathcal{X})$
\end{algorithmic}
\end{algorithm}

After obtaining the $MIS(\mathcal{X})$, we can use Algorithm~\ref{alg:center-visit} to generate a path that visits the center of each region. Since we have a lower bound on the region curvature, it is possible to construct a finite trajectory that visits all overlapping regions of $MIS(\mathcal{X})$.

\subsubsection{Detour}\label{sec:detour}
\begin{figure}
    \centering
    \includegraphics[width=0.4\textwidth]{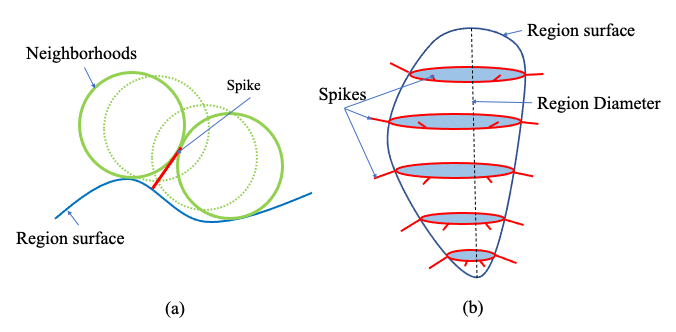}
    \caption{(a) In 2D, a detour (spike) can visit all neighborhood disks that touch the region surface. (b) In 3D, a detour (red curves) can visit all neighborhoods that touch the region surface.  The spikes are line segments with distance $D_{min}$ that are parallel to the surface normal.}
    \label{fig:detour}
\end{figure}
Denote a set of overlapping regions to $R(x) \in MIS(\mathcal{X})$ as $K(x) = \{R(y_1), R(y_2), R(y_k)\}$ such that $R(x) \bigcap R(y_j) \neq 0, \ \forall j=[1,..,k]$. First, we find the line segment with two endpoints $a,b$ in $R(x)$ that correspond to the $D_{max}$ of $R(x)$.
We also define a plane $P_i$ that is perpendicular to $\overline{a, b}$ and pass through point $a + \frac{\overline{a, b}}{|\overline{a, b}|} D_{min} * i, \ i=[1,...,k-1]$, which is a equal space of $D_{min}$ along the $\overline{a, b}$ line segment. From point $a$ to $b$, we visit the curve created by the intersection between $R(x)$ and the plane $P_i$. For each visit of the perimeter at plane $P_i$, we generate a set of spikes in the same plane as shown in Figure~\ref{fig:detour} of length $D_{min}$ with $D_{min}$ spacing around the curve. Those spikes and the perimeters will guarantee to touch all possible overlapping regions in $K(x)$.

\begin{lemma}\label{lem:detour}
Given a set of overlapping regions $R(x) \in MIS(\mathcal{X})$ as $K(x) = \{R(y_1), R(y_2), R(y_k)\}$. There exists a detour of at most $\frac{3\pi  D_{max}^2}{D_{min}}$ that visits all regions in $K(x)$.
\end{lemma}
\begin{proof}
Since all the regions have a minimum diameter of $D_{min}$, the spikes on the perimeter will need to be at most $D_{min}/2$ centered at a point $c$ on $R(x)$ for both outward and inward to cover the neighborhoods. The neighborhoods that the spike covers will be all the overlapping regions that touch $R(x)$ within a circle of diameter $D_{min}$ centered at point $c$. The maximum length for each perimeter is $\pi D_{max}$ and the maximum length for all spikes on the same perimeter is $\frac{\pi D_{max}}{D_{min}} D_{min}$. Therefore, the total detour length for $R(x_1)$ is at most
$\frac{D_{max}}{D_{min}}(\pi D_{max} + 2 \frac{\pi D_{max}}{D_{min}} D_{min}) = \frac{3\pi D_{max}^2}{D_{min}}$

\end{proof}

Given Lemma~\ref{lem:detour} and Theorem~\ref{thrm:center-visit}, we can present the proof for Theorem~\ref{thrm:detour-visit}.

\begin{proof}
There are at most $N$ detours for each region in the $MST(\mathcal{X})$. Similar to Theorem~\ref{thrm:center-visit}, we can obtain the final approximation factor $|J_d| \leq (1+\epsilon)(1 + \frac{27D_{max}}{40D_{min}} + \frac{81\pi D_{max}^2}{20D_{min}^2}) |E[J^*]|$
\end{proof}

\section{Online disjoint case}
When visiting for a set of objects $\mathcal{X}$ with only prior knowledge of their locations, we cannot orientate their regions $R(\mathcal{X})$.
The only information is the maximum and minimum detection range, which corresponds to $D_{max}/2$ and $D_{min}/2$. Therefore, the region shape for each object changes to a hollow ball. 

\begin{theorem}
Given a set of disjoint spheres in 3D with diameters bounded between $D_{max}$ and $D_{min}$, the length of the expected trajectory $J$ from Algorithm~\ref{alg:center-visit} is at most $O(\frac{D_{max}}{D_{min}})$ of $|E[J^*]|$.
\end{theorem}

During the online detection process, we do not know the actual diameter of the detection range other than the maximum and minimum range as $D_{max}$ and $D_{min}$. Using Lemma~\ref{lem:minkowski} will introduce a large ratio due to the uncertainty of the detection range. Therefore, we borrow the idea from Tekas et. al.~\cite{tekdas2012efficient} to bound the number $N$ w.r.t. the trajectory length. From Theorem 1 in~\cite{tekdas2012efficient}, they provide a bound that explore the minimum traveling cost $OPT$ among $N$ identical non-overlapping disks of diameter $D$ as
\begin{equation}\label{eq:n_bound}
    \frac{1}{4}N\alpha D \leq OPT
\end{equation}
where, $\alpha = 0.4786$. The same strategy works in 3D as show in Figure~\ref{fig:sphere_intersection} by intersecting a plane that connects the 3 sphere centers that converts the problem to the original 2D version. 

\begin{figure}
    \centering
    \includegraphics[width=0.4\textwidth]{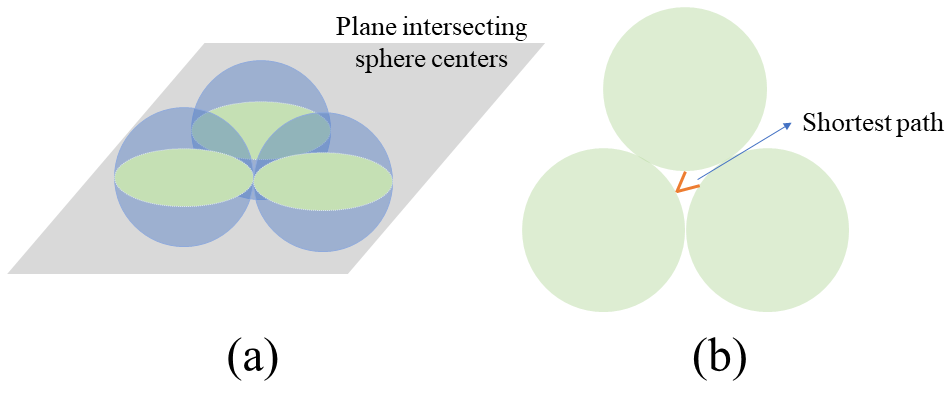}
    \caption{$(a)$~Given 3 spheres of the same radius $r$, we can intersect a plane through their spherical centers. $(b)$~The resulting intersection between the spheres and the plane are 3 disks of the same radius. The shortest path through those spheres is $0.4876 r$~\cite{tekdas2012efficient}.}
    \label{fig:sphere_intersection}
\end{figure}

\begin{proof}
Assume that the optimal trajectory for the spheres centers are $J^*_{c}$, which is bounded by the following.
$|J^*_{c}| \leq |J^*_{out}| + ND_{max}$
where $J^*_{out}$ is the optimal trajectory that visits all spheres with $D_{max}$.
We can plug in the upper bound for $N$ using Eq~\ref{eq:n_bound} so that the equation becomes
$|J_{c}| \leq (1+\epsilon)(1 + \frac{4D_{max}}{\alpha D_{min}})) E[|J^*_{out}|]
$
\end{proof}

\section{Experiments}\label{sec:expr}
\begin{figure}[th!]
    \centering
	\includegraphics[width=0.45\textwidth]{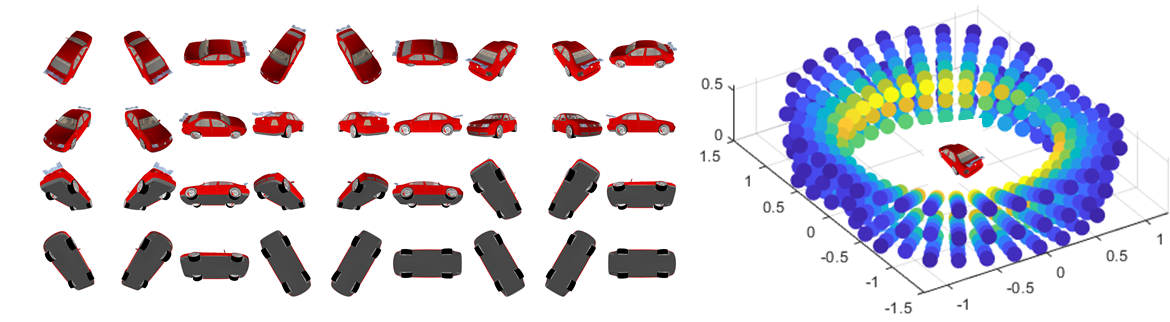}
    \caption{Omnidirectional view of a car and their corresponding scores. We applied a threshold to limit the views. The rest views formed a diameter-bounded region. The object orientations are shown in the center of the viewing scores. (High entropy to low entropy corresponds blue to yellow view points.)}
    \label{fig:region_shape}
\end{figure}
\begin{figure}[th!]
    \centering
	\includegraphics[width=0.46\textwidth]{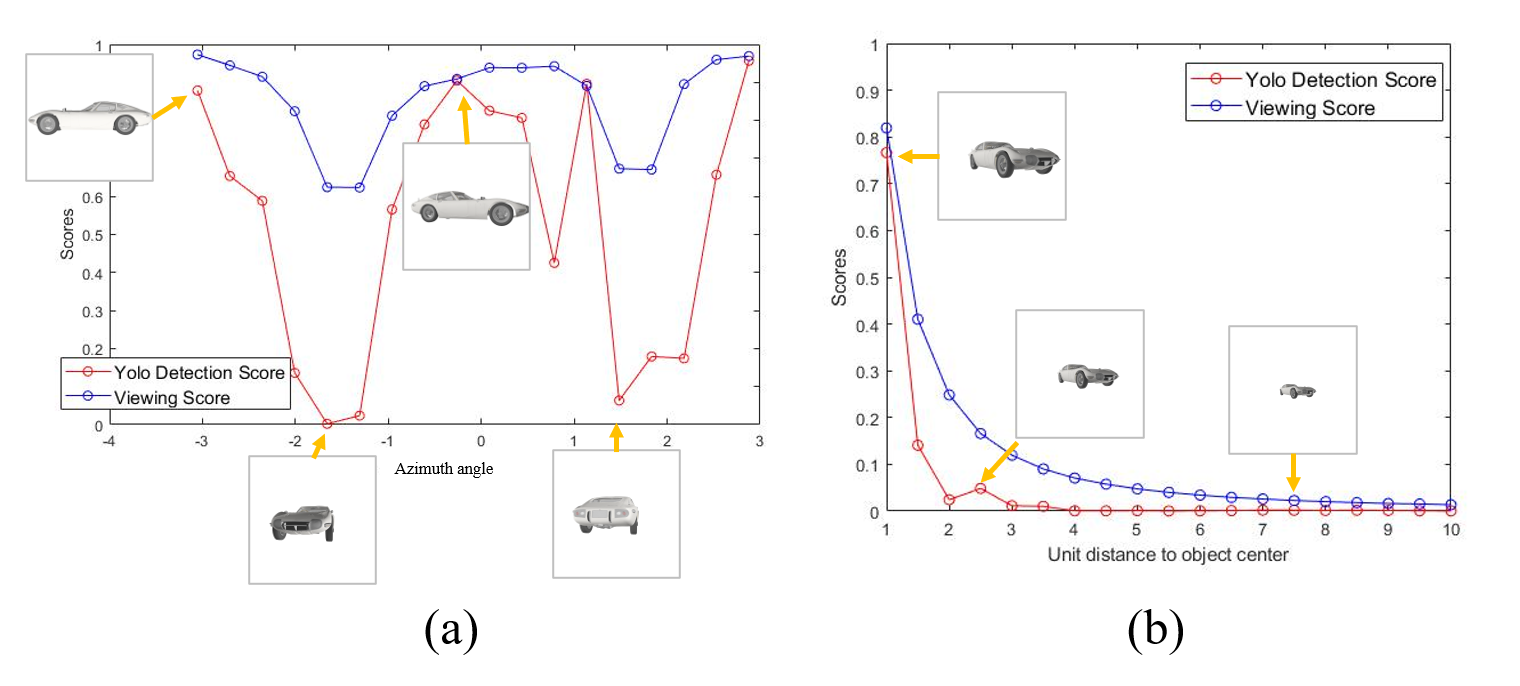}
    \caption{Comparison between yolo detection and our viewing scores. $(a)$~Scores for images with the same elevation angle and azimuth angle ranging from $-\pi$ to $\pi$. $(b)$~Scores for images of the same elevation and azimuth angle with increasing viewing distance to the center of the car. }
    \label{fig:car_detection}
\end{figure}
In this section, we first validate our assumption of the diameter-bounded region by evaluating omnidirectional views from various objects' 3D models. Then we compare the trajectory length and computation time with a method  proposed by Elbassioni et. al.~\cite{elbassioni2009approximation} for fat objects. Tests are done using a desktop with i7-9700 CPU and NVIDIA 2080 GPU. 

\subsection{Entropy-based region construction}

\begin{table}[th!]
    \centering
    \begin{tabular}{c|c|c|c|c|c|c|c}
        \toprule
        \bottomrule
        \multicolumn{2}{c|}{}& Car & Bus & Piano &  Table & Chair & Bed \\
        \toprule
        \bottomrule
        \multirow{4}{*}
        & Mean Maximum & 8.2 & 17.3 & 6.2 & 5.6 &  3.4 & 5.2   \\
        & Mean Minimum & 5.4 & 13.4 & 4.5 & 3.3 &  1.3 & 3.2   \\
        & Unit 1 distance & 3.5 & 10.4 & 3.1 & 2.5 &  0.5  & 2.0   \\
        \bottomrule
    \end{tabular}
    \caption{The average maximum and minimum viewing/region diameter (m) for each object by applying a score threshold of $0.3$. Unit 1 distance is the minimum radius that views are taken from the center of the object.}
    \label{tab:diameters}
\end{table}

\begin{table*}[th!]
    \centering
    \footnotesize
    \begin{tabular}{c|c|c|c|c||c|c|c||c|c|c}
        \toprule
        \multicolumn{2}{c|}{}& \multicolumn{3}{c||}{100 objects} & \multicolumn{3}{c||}{250 objects} & \multicolumn{3}{c}{500 objects} \\
        \cline{3-11}
        \multicolumn{2}{c|}{}& mean $|J|$ & STD &mean time&  mean $|J|$  & STD & mean time  &mean $|J|$ & STD &mean time\\
        \toprule
        \bottomrule
        \multirow{2}{*}{Car}
        & $\alpha$-fat &  3.1295 & 0.285 & 0.15 & 6.835 & 0.805& 1.12& 12.809 & 0.688& 7.521\\
        & \textbf{Ours}& 1.5461 & 0.073 & 0.03 & 2.692 & 0.077 & 0.26& 4.129 & 0.087 & 1.644 \\
        \midrule
        \multirow{2}{*}{Bus}
        & $\alpha$-fat & 2.7693 & 0.367 & 0.16 & 5.612 & 0.821& 1.13& 9.888 & 0.642& 7.515\\
        & \textbf{Ours}& 1.3284 & 0.064 & 0.04 & 2.225 & 0.085 & 0.27& 3.191 & 0.079 & 1.635 \\
        \midrule
        \multirow{2}{*}{Chair}
        & $\alpha$-fat & 3.2569 & 0.429 & 0.15 & 7.109 & 0.702& 1.12& 12.422 & 0.716& 7.340\\
        & \textbf{Ours}& 1.6023 & 0.071 & 0.03 & 2.788 & 0.075 & 0.27& 4.350 & 0.050 & 1.626 \\
        \midrule
        \multirow{2}{*}{Bed}
        & $\alpha$-fat & 3.3159 & 0.283 & 0.16 & 7.104 & 0.901& 1.14 & 13.498 & 1.292& 7.012\\
        & \textbf{Ours}& 1.6167 & 0.049 & 0.03 & 2.856 & 0.068 & 0.27& 4.471 & 0.074 & 1.519 \\
        \midrule
        \multirow{2}{*}{Table}
        & $\alpha$-fat & 3.3569 & 0.399 & 0.14 & 7.167 & 0769& 1.13 & 15.783 & 0.947& 7.081\\
        & \textbf{Ours}& 1.6482 & 0.065 & 0.04 & 2.905 & 0.076 & 0.27& 4.455 & 0.068 & 1.610 \\
        \midrule
        \bottomrule
    \end{tabular}
    \caption{A comparison of the trajectory length ($10^3 m$) $|J|$ and the run-time (seconds) between the $\alpha$-fat method and ours.}
    \label{tab:length_runtime_comparison}
\end{table*}

\begin{figure*}[t!]
    \centering
	\includegraphics[width=0.6\textwidth]{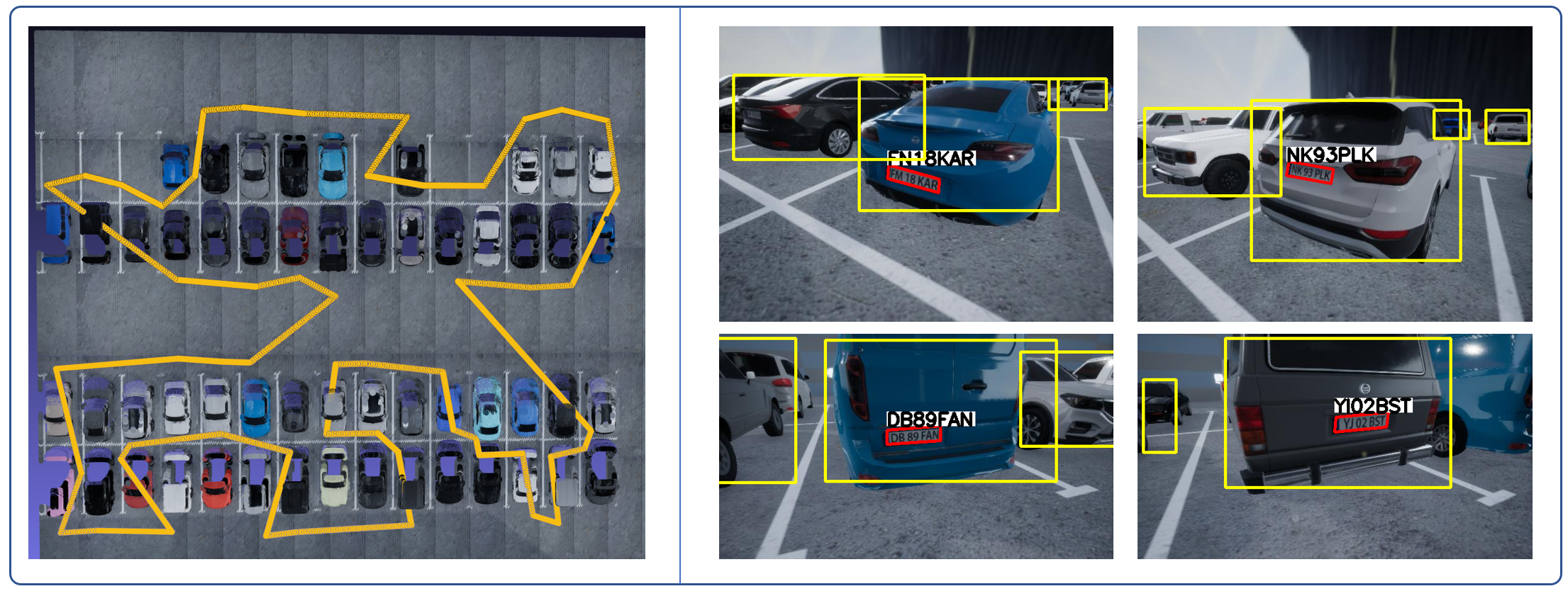}
    \caption{Left: Top-down view of the trajectory planned for an aerial vehicle to detect all license plates of all vehicles in a parking lot. Right: Four examples of vehicle license plates detection and recognition using ALPR-unconstrained network~\cite{silva2018a}}
    \label{fig:license_plates_detection}
\end{figure*}
Given a set of views for an object, a pre-trained neural network can output detection scores for each view. However, the current training dataset (such as COCO~\cite{lin2014microsoft}) is collected mostly from canonical views. Due to the lack of sufficient views from non-canonical perspectives and unbalanced samples, it is biased to generalize viewing scores using a neural network approach. Therefore, we propose a viewing score that allows the differentiation of objects from different views and thus approximates their corresponding detectability. Our proposed viewing scores use the entropy of image edge orientations to distinguish views from different perspectives. Since edge orientation distribution does not change much while changing the viewing distance, we also introduce an object-to-image ratio. 
The viewing score $S$ is given as
\begin{equation}\label{eq:viewing_score}
    S =(-\sum P *\log{P}) * R 
\end{equation}
where $P$ is the probability of the edge orientation distribution per pixel quantized into 360 bins, $-\sum P *\log{P}$ is the image edge entropy, and $R = \frac{Num \ of \  object \  pixels}{Image \ Area}$ is the object to image area ratios. When an object is further away or the image edge is biased towards a small set of orientations, $S$ will be small.

A similar idea was proposed by~\cite{hrvzic2019local} to avoid biases from the training data set and network architecture. We evaluated 6 different types of models from Shapenet~\cite{shapenet2015} that contain cars, buses, cups, tables, chairs, and beds. For each category, we obtained 1080 images for 10 different models. 

We compute the scores for images and threshold the entropy value to $0.3$ and the resulting region shape is shown in Figure~\ref{fig:region_shape}. 
To compare with actual detection scores, we evaluate our `Shapenet` images using a pre-trained yolo network~\cite{yolov3}. 
For more canonical views, there is a high correlation between the detection scores and our viewing score $S$ as shown in Figure~\ref{fig:car_detection}. 

For all the models, we evaluate the expected maximum and minimum region diameter for the given threshold, which are shown in Table~\ref{tab:diameters}. It is clear that for objects like cars, both diameters are very similar. However, for objects such as chairs and beds, the information is much more limited from the side with a smaller cross-section area. It is because those objects have very small footprints from the side and thus have fewer varieties of edges presented in the image. 

\subsection{Trajectory comparison}
We also test our algorithm in simulation to show that the proposed method is efficient in length and computational cost. The work proposed by Elbassioni et. al.~\cite{elbassioni2009approximation} obtains an approximation factor of $9.1\alpha + 1$ for disjoint $\alpha$ fat regions, where $\alpha$ is defined as a measure of fatness and it has a similar problem set up comparing to our proposed diameter-bounded region problem. Therefore, we use this method as the baseline to compare ours with. We will denote Elbassioni et. al.~\cite{elbassioni2009approximation}'s method as $\alpha$-fat.
$\alpha$-fat greedily selects a representative point for each neighborhood that is close to each other. The trajectory is planned using a TSP solver~\cite{concorde}. 

For online cases, we compare the trajectory lengths and computation time for each method with randomly selected 100, 250, and 500 locations within a 3D cube with a 100-meter edge length. For each location, there is a $D_{max}$ and $D_{min}$. We sample 108 points on each sphere with diameter $D_{max}$ for the $\alpha$-fat method. Since $D_{max}$ and $D_{min}$ are only a bound for the ground truth diameter for each region, we sample ground truth diameters uniformly at random between the bounds. To calculate the trajectory through a set of 3D locations, we also use the Concord TSP solver~\cite{concorde}. The comparison in Table~\ref{tab:length_runtime_comparison} shows that the trajectory length and run-time for $\alpha$-fat method are much higher compared to our method. It is because each neighborhood needs to be represented by a set of surface points (108 points) and the computation is affected by the density of those points.

\subsection{Trajectory comparison for car detection}
We show that our method can be implemented efficiently in a photo-realistic environment (Unreal Engine 4~\cite{unrealengine} with city model~\cite{city}) to detect various objects in real scenes. We randomly placed 30 cars from Shapenet models~\cite{shapenet2015} in the scene. We use the `yolo' detection network~\cite{yolov3} for each image captured along the trajectory. Since the detection is online, we picked the detection range for all the objects based on Table~\ref{tab:diameters} and a set of predetermined locations. As shown in Figure~\ref{fig:car_detection}, our trajectory detects all the cars in the scene. To avoid obstacles such as buildings in the scene, we implemented RRT*~\cite{islam2012rrt} to adjust the trajectory. 

\subsection{Trajectory planning for license plate extraction}

The previous experiment assumes known vehicle locations. If vehicle locations are given, it is usually not necessary to plan for general-purpose detections again. However, if more specific attributes, such as vehicle brand, color, and license plate numbers, are required, an efficient trajectory is desired. In this experiment, we show that our algorithm can efficiently detect and recognize the license plates for all vehicles in a parking lot. We choose a parking lot environment~\cite{parkinglot} as the testing cases from Unreal Engine. This environment contains 54 different vehicle with license plates, including van, SUV, hatchback, sedan, pickup, etc. License plates are only on the back of each vehicle since not all states in the U.S. require front plates. Each vehicle's location and orientation is given in advance. The vehicle models are also given to infer license plate locations. For each license plate, we impose a viewing score of 0.6 using Eq~\ref{eq:viewing_score}. For plate detection and extraction, we use a pre-trained network `ALPR-unconstrained'~\cite{silva2018a}. 
Given each vehicle's location, orientation, and  viewing score, we can construct a diameter-bounded region for each license plate. The resulting trajectory that detects all the license plates for all vehicles in the parking lot is shown in Figure~\ref{fig:license_plates_detection}. We also ran detection on each license plate, which successfully recognized most of the license plates. The resulting detection is shown in Figure~\ref{fig:license_plates_detection} with an average detection score of $84.3\%$, which is far above the score threshold ($0.3$) for obtaining the detection region.

\section{Conclusion}
We studied the problem of planning the shortest tour and viewing locations for objects in a scene. Since the detection score cannot be predicted without actually seeing the object, we correlate the viewing pose to the detection score distribution. Since a pre-trained network may suffer from biased datasets and training architectures, we use an entropy-based viewing score to capture detectability distribution. Such regions are diameter-bounded and can be non-convex. We presented an efficient method to find a tour that visits all of them and detects all objects of interest. Our formulation solves for both disjoint and non-disjoint cases, with the approximation ratios of $O(\frac{D_{max}}{D_{min}})$ and $O(\frac{D_{max}^2}{D_{min}^2})$, respectively.

Our method requires object locations in prior. A possible future direction for unknown locations is to first plan an exploration step to extract the locations.




\bibliographystyle{IEEEtran}
\bibliography{IEEEabrv,egbib}
\end{document}